\documentclass[11pt]{article}
\usepackage[margin=1in]{geometry}
\usepackage[english]{babel}
\usepackage{xr}
\usepackage{booktabs}
\usepackage{setspace}
\usepackage{cancel}
\usepackage[normalem]{ulem}
\usepackage{graphicx, xcolor}

\usepackage{tikz}
\usetikzlibrary{arrows,automata,positioning}

\PassOptionsToPackage{hyphens}{url}
\usepackage{mathrsfs,hyperref}
\usepackage[]{amsmath}
\usepackage{amsthm}
\usepackage{fix-cm}
\usepackage[]{amssymb}
\usepackage[]{latexsym}
\usepackage[right]{eurosym}
\usepackage[T1]{fontenc}
\usepackage[]{graphicx}
\usepackage[]{epsfig}
\usepackage{fancyhdr}
\usepackage{pstricks}
\usepackage{multirow}
\usepackage{subcaption}
\usepackage[inline]{enumitem}

\allowdisplaybreaks

\makeatletter

\newcommand{\bbr}{\mathbb{R}}
\newcommand{\E}{\mathbb{E}}

\newcommand{\bbn}{\N_0}

\newcommand{\bbp}{\mathbb{P}}

\newcommand{\bbf}{\mathbb{F}}

\newcommand{\om}{\omega}

\newcommand{\fcal}{\mathcal{F}}

\newcommand{\dcal}{\mathcal{D}}

\newcommand{\bcal}{\mathcal{B}}

\usepackage{cite}
\def \abs#1{\left| #1 \right| }

\newcommand{\N}{\mathbb{N}}

\newcounter{modcount}
\newcommand{\modulo}[2]{%
\setcounter{modcount}{#1}\relax
\ifnum\value{modcount}<#2\relax
\else\relax
\addtocounter{modcount}{-#2}\relax
\modulo{\value{modcount}}{#2}\relax
\fi}
\newcommand{\tablepictures}[4][c]{\begin{tabular}[#1]{@{}c@{}}#2\vspace{0.5cm}\\(\alph{#4}) #3\end{tabular}}
\newcounter{gridsearch}
\newcommand{\tabpic}[2]{
    \stepcounter{gridsearch}
    \modulo{\thegridsearch}{2}
    \ifnum\value{modcount}=0
        \tablepictures[t]{#1}{#2}{gridsearch}\\[2.0cm]
    \else
        \tablepictures[t]{#1}{#2}{gridsearch}&~&
    \fi
}

\makeatother
\newtheorem{lemma}{Lemma}[section]
\newtheorem{proposition}[lemma]{Proposition}

\newtheorem{corollary}[lemma]{Corollary}

\newtheorem{example1}[lemma]{Example}
\newtheorem{rem1}[lemma]{Remark}
\newtheorem{assumption}[lemma]{Assumption}
\newtheorem{alg1}[lemma]{Algorithm}
\newtheorem{me1}[lemma]{Mechanism}

\newenvironment{remark}{\begin{rem1}\rm}{\end{rem1}}
\newenvironment{example}{\begin{example1}\rm}{\end{example1}}

\newenvironment{algorithm}{\begin{alg1}\rm}{\end{alg1}}

\usepackage{color}

\newcommand{\T}{\top}

\DeclareMathOperator*{\argmax}{arg\,max}

\newcommand\ind[1]{\mathbb{I}_{\{#1\}}}

\newcommand\bp{\bold{p}}
\newcommand\bb{\bold{b}}
\newcommand\bV{\bold{V}}
\newcommand\bK{\bold{K}}
\newcommand\bdf{\bold{f}}
\newcommand\bx{\bold{x}}
\newcommand\bL{\bold{L}}
\newcommand\bw{\bold{w}}
\newcommand\bd{\bold{d}}
\begin{document}

\title{\bf Decentralized Payment Clearing using Blockchain and Optimal Bidding}
\author{Hamed Amini \thanks{Robinson College of Business, Georgia State University, Atlanta, GA 30303, USA, email: {\tt hamini@gsu.edu}} \and
Maxim Bichuch \thanks{Department of Applied Mathematics and Statistics, Johns Hopkins University, Baltimore, MD 21218, USA, email: {\tt mbichuch@jhu.edu}. Work  is partially supported by NSF grant DMS-1736414. Research is partially supported by the Acheson J. Duncan Fund for the Advancement of Research in Statistics.} \and
Zachary Feinstein \thanks{Stevens Institute of Technology, School of Business, Hoboken, NJ 07030, USA. \tt{zfeinste@stevens.edu}}}
\date{\today}
\maketitle
\abstract{
In this paper, we construct a decentralized clearing mechanism which endogenously and automatically provides a claims resolution procedure.  This mechanism can be used to clear a network of obligations through blockchain.  In particular, we investigate default contagion in a network of smart contracts cleared through blockchain.  In so doing, we provide an algorithm which constructs the blockchain so as to guarantee the payments can be verified and the miners earn a fee.  We, additionally, consider the special case in which the blocks have unbounded capacity to provide a simple equilibrium clearing condition for the terminal net worths; existence and uniqueness are proven for this system.  Finally, we consider the optimal bidding strategies for each firm in the network so that all firms are utility maximizers with respect to their terminal wealths. We first look for a mixed Nash equilibrium bidding strategies, and then also consider Pareto optimal bidding strategies.
 The implications of these strategies, and more broadly blockchain, on systemic risk are considered.  

\bigskip

}

\newpage
\section{Introduction}\label{sec:intro}

Clearing mechanisms and default resolution procedures are fundamental to the functioning of the financial system.  Through such mechanisms, payments of obligations are made.  If a claim is defaulted upon, the obligee will record losses and may default on its own liabilities.  Such a chain of defaults is referred to as default contagion.  This is a specific form of systemic risk; the study of which has proliferated since the 2007-2009 financial crisis due to the large economic costs of that event.  In the last few years, decentralized finance (i.e., the replacement of traditional financial intermediaries with blockchain-based systems) has grown from the foundations of Satoshi Nakamoto's whitepaper proposing a decentralized cash system~\cite{bitcoin}.  In this paper, we consider a decentralized clearing system which endogenizes the default resolution process so as to automate the payment of obligations and resolve disputes.

The rest of the paper is organized as follows.  We complete this introduction with a review of relevant literature in Section~\ref{sec:intro-lit}, provide specific motivations for blockchain clearing systems in Section~\ref{sec:intro-motivation}, and highlight the key contributions of this work in Section~\ref{sec:intro-contrib}.  We provide a broad overview of the financial scope of our work within Section~\ref{sec:overview}. In Section~\ref{sec:centralized}, we present the model of financial network and default contagion in a centralized clearing system as introduced in the Eisenberg-Noe framework~\cite{EN01}. We further consider the clearing  with bankruptcy costs as in~\cite{RV13}  and collateral as presented in~\cite{veraart-compression,amini-compression}. In Section~\ref{sec:decentralized}, we consider a decentralized clearing system where, instead of a single central authority, a blockchain system for clearing is considered which is supported by a set of \emph{external} miners. In Section~\ref{sec:bidding}, we study the formulation of the optimal bidding strategies for each firm in the financial network as the result of a non-cooperative game and as a Pareto optimal strategy. Section~\ref{sec:conclusion} concludes.

\subsection{Literature Review}\label{sec:intro-lit}

Since the financial crisis 2007-2009, there has been an extensive research to study clearing systems and to understand how an economic shock propagates in the financial system through a variety of interconnections, such as interbank payments~\cite{EN01,RV13}, cross-holdings~\cite{EGJ14,AW_15,E07}, balance sheet-exposures~\cite{ACP14,BEST:2004,GK10,NYYA07}, and common asset holdings~\cite{chen2014asset,caccioli2014stability}. In particular, the Eisenberg-Noe clearing framework~\cite{EN01}  provides a simple versatile equilibrium model  in which, starting from a network of nominal liabilities, one assumes a proportional sharing rule (in the case of a default) and obtains in equilibrium a cleared network of actual payments. This model has been extended in different directions to cover new features. To name just a few, \cite{RV13,GY14} introduces bankruptcy costs and mark to market losses, \cite{AFM16,bichuch2018borrowing,feinstein2015illiquid,CFS05} consider the fire sales for a single and multiple illiquid assets. 


A crucial assumption in the Eisenberg-Noe framework~\cite{EN01} is the \emph{pro-rata repayment rule} (all obligations are repaid in proportion to the size of the nominal claims). As shown in~\cite{csoka2018decentralized,schaarsberg2018solving,ketelaars2020decentralization}, the proportionality rule is not critical and other bankruptcy rulesets could be utilized instead. In particular,~\cite{csoka2021axiomatization} provides an axiomatization of the proportional rule in financial networks. In~\cite{csoka2018decentralized}, the authors consider a discrete model for a class of decentralized clearing processes and  show that, under a sufficiently small unit of account, the clearing equities converge to the same values as the centralized clearing procedure. \cite{ketelaars2020decentralization} builds on the work of~\cite{schaarsberg2018solving} (which can be regarded as a centralized mechanism) and show that firms can reach a consensus on how to allocate the total estate when using a specific individual settlement allocation procedure ($\phi$-based mutual liability rules, in contrast to pro-rata assumption as in \cite{EN01}).


\subsection{Motivation}\label{sec:intro-motivation}

Decentralized finance, i.e.\ the replacement of traditional financial intermediaries with blockchain-based systems, is rapidly growing in size and scope.  The total value locked up (in dollar denominated terms) in decentralized finance has grown roughly 9,000\% between January 1, 2020 and January 1, 2022.\footnote{\url{https://defipulse.com/}}  In such systems, traditional financial instruments are replaced with smart contracts that are self-executing.  Broadly speaking, in this paper, we study clearing smart contracts in a blockchain-based framework.

One advantage of decentralized systems is their resilience to certain forms of systemic breakdown -- an understudied area of systemic risk within the financial engineering community.  In any system with centralized intermediaries, there is the risk of catastrophic failure of those firms; see, e.g.,~\cite{bignon2020failure} for an empirical study analyzing the failure of a derivatives central clearing counterparty (CCP) in Paris in 1974 and \cite{duffie2015resolution} for discussion on the resolution of failing CCPs. The blockchain-based distributed ledger mitigates this risk by the nature of it decentralization.   
Therefore, the use of a blockchain-based system can inherently increase financial stability by mitigating a form of catastrophic systemic failure. 

As such, a decentralized clearing system is of the utmost importance.  However, to the best of our knowledge, no such blockchain-based mechanisms has been rigorously studied previously.  Such a clearing mechanism needs to incorporate considerations for default resolution of smart contract claims.  (At the time of this writing, the ISDA Master Agreement for smart contracts\footnote{\url{https://www.isda.org/a/23iME/Legal-Guidelines-for-Smart-Derivatives-Contracts-ISDA-Master-Agreement.pdf}} does not provide specific guidance for the treatment of defaults on smart contracts which can cause disputes due to the self-executing nature of smart contracts.)
Notably, the primary difficulties in claims resolution is determining the available funds and disagreements over seniority of payments.  These notions are modeled in a stylized way in the centrally cleared frameworks of, e.g.,~\cite{EN01,RV13} by applying a fixed recovery rate to the entire payment following a proportional repayment scheme.
A blockchain clearing mechanism, on the other hand, resolves these two issues automatically. As such, the lack of default resolution encoded in the ISDA agreement for smart contracts does not pose a particular hindrance for consideration of smart contracts as a closed system.
This reduces frictions and speeds up the clearing procedure as the blockchain system automates the claims resolution procedure through the seniority structure determined by bids made by the obligees.  This occurs entirely endogenously without introducing additional constraints, e.g., the pro-rata repayment scheme.  In allowing institutions to (decentrally) determine their desired seniority, the blockchain clearing system endogenizes the effective recovery rate.

We wish to remark that other decentralized clearing procedures have been proposed (see, e.g.,~\cite{schaarsberg2018solving,ketelaars2020decentralization}), but none (to the best of our knowledge) have specifically studied the implications of the blockchain and distributed ledger on the actualized payments.  In fact, we find that explicitly modeling the blockchain allows for the selected claims resolution (herein a bidding strategy to specify the seniority structure and fees paid to the blockchain miners) to be determined in an endogenous manner.

\subsection{Primary contributions}\label{sec:intro-contrib}

The primary contributions of this paper are as follows:
\begin{enumerate}
\item 
We formulate a novel decentralized \emph{blockchain clearing mechanism} in Section~\ref{sec:blockchain}. This procedure is presented from the viewpoint of the miners who verify transactions and construct the blocks to be placed on the distributed ledger.  The fundamental elements of this clearing mechanism are the bids that accompany any obliged payment. These bids specify the fees associated with a nominal amount owed.  The miner who places the block containing those obligations into the ledger collects the fees.  Because miners wish to maximize their own profits, the fees act as a seniority structure on the obligations with larger bids more likely to be fulfilled (provided a default occurs). It is assumed that the fees are being subtracted from the payments, and therefore part of the verification of the bids includes the verification that the payment obligation will be paid.
In this way, this blockchain clearing mechanism is a decentralized clearing system that automates the claims resolution procedure and guarantees that the prescribed rules are followed.
In fact, we provide an algorithm (Algorithm~\ref{alg:bids}) that constructs the entire blockchain to clear these obligations given the bids placed for these obligations. In turn, given this construction of the blockchain, we also consider and solve the banks' optimization problems, of how to maximize their own utility and find the optimal bidding strategy. 
\item 
We, further, consider the formulation of the \emph{optimal bidding strategies} for each firm in the network and study the implications of these strategies, and more broadly blockchain, on systemic risk in Section~\ref{sec:bidding}. We show how the optimal bidding strategy for one obligation may depend on the bids submitted for all other obligations.  We then consider two distinct approaches to forming the optimal bidding strategies: \begin{enumerate*} \item as the result of a complex non-cooperative game; and \item as a Pareto optimal strategy. \end{enumerate*} We show that a pure strategy equilibrium need not exist, while there always exists a mixed strategy Nash equilibrium bidding strategy for any blockchain clearing system with finite discretization. We also show that  there exists a Pareto optimal bidding strategy for any blockchain clearing system with finite discretization. By construction, these optimal strategies provide a heuristic (firm-level) metric of systemic risk.  Specifically, if the system is unstressed then all bids will be optimized at the minimal fee levels. Whereas, under stresses the firms will need to raise their fees so as to optimize the seniority structure of payments.  Thus the larger the optimal fees, the larger the expected stresses to the system.
\item 
We expand on the blockchain clearing mechanism to consider the \emph{terminal net worths} (i.e., realized assets minus nominal liabilities) for all institutions in the financial system following the prescribed bidding strategy.  In particular, if the blocks have sufficiently large capacity, we find an equilibrium formulation satisfied by these terminal net worths.  This equilibrium formulation can be viewed comparably to the Eisenberg-Noe framework with seniority as studied in, e.g.,~\cite{E07}; however, the system presented herein endogenizes the costs for the seniority structure which, as far as the authors are aware, has not previously been studied.  For this equilibrium formulation of the blockchain system, we show that under weak regularity conditions there is a unique clearing net worths vector and, generally, there always exist greatest and least clearing net worths.  This equilibrium formulation provides an efficient computational procedure which can be used to approximate the blockchain clearing solution. 
\end{enumerate}

\section{Overview of Financial Setting}\label{sec:overview}
The setup for our paper assumes a network of banks with liabilities. We assume that the liabilities might be owed in physical or on-chain assets. At some point, an event occurs that requires an immediate liquidation of the securities. Such an event can be a shock to the banking system, the eventual repayment expiration or maturity time of the liabilities or assets, or the time when not all banks are able to fully repay their liabilities. Whatever the cause, the initial setup is that a network of interbank liabilities needs to be cleared. 

The typical practice so far in the event of nonpayment has been to declare bankruptcy and let the courts sort it out. This is a time consuming and costly solution. Herein we propose to let the banks themselves sort it out, using a decentralized clearing mechanism. The first step in this process is to tokenize all the assets in the system (if not already tokenized). Our model is semi-static insofar as the system evolves over time but is determined by decisions made prior to the clearing time. In any case, the procedure itself should be very fast.  As such, even though the assets must be frozen during the clearing process, the banks will have access to the resulting cleared positions soon thereafter (e.g., at the start of the next business day with overnight clearing). 

The key idea of our decentralized clearing algorithm is the bidding. All banks must bid on the assets they are owed. Realistically the bids are discrete -- though for notational convenience we will model them as continuous bids -- with the fee being proportional to the amount owned.\footnote{It is trivial to accomodate absolute fees instead.  Such a modification results in a more cumbersome formulation of the clearing problem.} Once the bids are placed, the miners clear the debts in a systematic manner. Specifically, the miners verify that a bank has sufficient funds to discharge the debt, in which case the miners place the bid on the block, and pocket the fee. It is important to note, that the fee is paid from the debt being processed. This aligns the incentives of the miners with the verification task; specifically, verification of bank funds to pay and discharge liabilities also guarantees the payment of the fees to the miner.
The block might have limits (e.g., a limit on the number of transactions processed within it), in which case the miners seek to maximize the payoff from the fees that they collect given the availability of funds.

We wish to highlight that the proposed algorithm is safe, since at all times either the bank has its assets or the appropriate part of the liability has been discharged once the asset was transferred. As is common in the decentralized universe, custodial services are not required. We also assume for convenience that once a block has been constructed, and verified, it has also been confirmed, though it is possible to relax this assumption to consider the unspent funds for the next few block(s), i.e.\ until final confirmation.

\section{Default Contagion in Centralized Clearing Systems}\label{sec:centralized}

Consider a financial system comprised of $n$ banks and set $[n] := \{1,\dots, n\}$. 
Throughout this work we will assume the stylized balance sheet considered in the Eisenberg-Noe framework~\cite{EN01}.  As such,  each bank $i$ has some liquid assets (cash) $x_i \geq 0$ and interbank assets $\sum_{j\in [n]} L_{ji}$ where bank $j$ owes $L_{ji} \geq 0$ to bank $i$.  On the other side of the balance sheet, bank $i$ has interbank liabilities $\sum_{j\in [n]} L_{ij}$.\footnote{ External obligations can be included by introducing a ``fictitious bank'' $0$ to the financial system.  This is explicitly introduced within Example~\ref{ex:pareto} later in this work.}
Within this setting no bank has any obligation to itself, i.e., $L_{ii} = 0$ for every bank $i$.  
Additionally, for ease of notation, define the relative liabilities from bank $i$ to $j$ as $\pi_{ij} := L_{ij} / \sum_{k \in [n]} L_{ik}$ if $\sum_{k \in [n]} L_{ik} > 0$ and $\pi_{ij} = 0$ otherwise.
Herein, and comparably to the models presented in~\cite{veraart-compression,amini-compression}, we consider a setting in which each of these interbank obligations is collateralized at level $\mu \in [0,1]$ and with bankruptcy costs defined by a recovery rate $\alpha \in [0,1]$.

Within this extended Eisenberg-Noe framework, the primary question is one of network clearing.  Specifically, how much does each bank pay towards its liabilities and what are the resulting net worths.  The (extended) Eisenberg-Noe system assumes the following financial rules.
\begin{enumerate}
\item\label{ll} \emph{Limited liabilities}: no firm pays more obligations than it has realized in assets. 
\item\label{pd} \emph{Priority of debt claims}: no equity is accumulated by the shareholders of a firm unless all firm debts are paid in full.
\item\label{cr} \emph{Collateral rehypothecation}: the margin account is fully available for the payment of debts.
\item\label{fa} \emph{Fractional recovery of assets}: in the event of a default, the collateral is recovered in full but all other assets are subject to the recovery rate.
\item\label{pr} \emph{Pro-rata repayment}: all obligations are repaid in proportion to the size of the nominal claims. 
\end{enumerate}
This is a \emph{centralized} clearing system due to the use of the overarching bankruptcy rules for defaulting firms (i.e., fractional recovery of assets) and the pro-rata repayment.  For instance, and specifically, pro-rata repayment only occurs following the dictates of some external authority which imposes such a rule. This external authority could be a clearinghouse or legal structure in bankruptcy courts.

Under this ruleset, the payments by bank $i$ follow the equilibrium clearing procedure 
\footnote{Note that the assumption on the collateral here differs slightly from the treatment in~\cite{veraart-compression,amini-compression}.  Herein for the purposes of comparison to the decentralized clearing procedure, since the assets are frozen and the process is (nearly) instantaneous, it is without loss of generality that we can assume that collateral is used immediately towards the repayment of obligations.} 
\begin{equation}\label{eq:centralized}
p_i = \Phi_i^*(\bp) := \mu\sum_{j \in [n]} L_{ij} + \begin{cases} (1-\mu)\sum_{j \in [n]} L_{ij} &\text{if } x_i + \sum_{j\in [n]} \pi_{ji}p_j \geq (1-\mu)\sum_{j \in [n]} L_{ij}, \\ \alpha\left[x_i + \sum_{j \in [n]} \pi_{ji}p_j\right] &\text{if } x_i + \sum_{j \in [n]} \pi_{ji}p_j < (1-\mu)\sum_{j \in [n]} L_{ij}. \end{cases}
\end{equation}
Notably, this system coincides with that of~\cite{RV13} if there is no collateralization (i.e., $\mu = 0$) and~\cite{EN01} if, additionally, there is full recovery in case of default (i.e., $\alpha = 1$).
\begin{proposition}\label{prop:centralized}
There exists a greatest and least clearing payments vector $\bp^\uparrow \geq \bp^\downarrow$ to the clearing procedure $\Phi^*$ defined in~\eqref{eq:centralized} for any collateralization level $\mu \in [0,1]$ and recovery rate $\alpha \in [0,1]$.
\end{proposition}
\begin{proof}
This result follows from a trivial application of Tarski's fixed point theorem on $\Phi^*: [0,\bar {\bp}] \to [0,\bar \bp]$, where $\bar{\bp}=(\bar{p}_1, \dots, \bar{p}_n)^\T$ and $\bar{p}_i:= \sum_{k \in [n]} L_{ik}$.
\end{proof}
Given a clearing payments vector $\bp=(p_1,\dots,p_n)^\T$, the resulting net worths $K_i$ for bank $i$ are:
\[K_i = x_i + \sum_{j \in [n]} \pi_{ji}p_j - \sum_{j \in [n]} L_{ij},\]
with cash account $V_i = K_i^+ := \max\{0,K_i\}$, the positive part of $K_i$. Note, that the net worth is defined as assets minus the liabilities, and can be negative. This accounts for the shortfall in case of default of bank $i$. This is opposed to the cash account, which is equal to the net worth if the latter is non negative, and otherwise the bank is in default and the cash is zero by definition.

\section{Default Contagion in Decentralized Clearing Systems}\label{sec:decentralized}


We now want to consider what happens in a decentralized system cleared by miners in a blockchain. We continue to assume the same financial network of banks with liabilities, each of these obligations is collateralized by $\mu \in [0,1]$ as presented in the centralized clearing system above. 

\subsection{Blockchain Clearing}\label{sec:blockchain}
The decentralized clearing systems of interest within this work continue to follow \begin{enumerate*} \item \emph{limited liabilities}, \item \emph{priority of debt claims}, and \item \emph{collateral rehypothecation}. \end{enumerate*}  
However, due to the decentralized nature of clearing under consideration herein, no authority can or will enforce bankruptcy laws which, allow for notions such as the fractional recovery of assets. Similarly, no authority can impose the pro-rata repayment scheme for debt repayment on the financial system.  

Instead of a single central authority, herein a blockchain system for clearing is considered, which is supported by a set of \emph{external} miners, i.e., we assume that no bank within this system also acts as a miner of the clearing blockchain.
Following the stylized structure of a blockchain system, all interbank assets $L_{ji},~ i,j\in [n],$ are paired with a bidding strategy $b_{ji}: [0,1] \to \bbr_+ \cup \{\infty\}$ that encodes the fees bank $i$ is willing to pay to a miner to place the obligation from bank $j$ onto the blockchain.  
More specifically, $b_{ji}(f)$ denotes the amount of the obligation from bank $j$ to $i$ for which a \emph{proportional} fee of $f$ is applied. 
Any bidding strategy must account for the entirety of the (uncollateralized) obligation, i.e., $\int_0^1 b_{ji}(f)df = (1-\mu)L_{ji}$.\footnote{Any bidding strategy with $\int_0^1 b_{ji}(f)df < (1-\mu)L_{ji}$ can be augmented with a mass of bids placed at the $0$ fee level without loss of generality.} In Section~\ref{sec:bidding} we will assume that the fees are discrete. However, for notational convenience we will write it using a continuous notation with an integral. The discrete bidding can then be recovered when $b_{ij}$ take infinite value, and we think of it as Dirac delta function.  To account for the possibility of discrete bidding, we will consider the notation \[\int_{(f^*)^+}^1 g(f)b_{ji}(f)df = \lim_{\epsilon \searrow 0} \int_{f^*+\epsilon}^1 g(f)b_{ji}(f)df\] for any bounded function $g: [0,1] \to [0,1]$ and any pair of banks $i,j \in [n]$.

This bidding structure completes the stylized rules for \emph{decentralized clearing} through a blockchain system.  In addition to rules~\eqref{ll}-\eqref{cr}, the clearing system satisfies the following rules.
\begin{enumerate}\setcounter{enumi}{3}
\item\label{pf} \emph{Post-payment fees}: all fees are paid from the incoming cash flows based on the bidding strategy. Therefore, the miners first verify that the payments can be made by the bank before putting it in the block. In other words, all payments on the block are simultaneously verified for the ledger to guarantee the miners get their fees.
\item\label{np} \emph{No payment rehypothecation}: by virtue of blockchain ledgers, all obligations must be paid with the current cash account and cannot immediately reuse an incoming payment to satisfy an obligation in the same block.\footnote{No payment rehypothecation comes from the idea that only verified transactions are considered valid. This assumption can be relaxed with the use of, e.g., unspent transaction outputs [UTXOs] in which all valid asset transfers are immediately accessible by the recipient.  We comment later in this work on how such a relaxation of this assumption impacts the mathematical modeling.}
\item\label{gb} \emph{Greedy block construction}: the miners will select the transactions with the highest fees to attach to the current block. This greedy approach is applied as each miner cannot guarantee it will ``win'' the consensus game on additional blocks.

\item\label{bs} \emph{Limited block capacity}: each block has a limited capacity $C^\# > 0$ on the number of transactions. 
\end{enumerate}

\begin{assumption}\label{assump:initial}
With rehypothecation, and an assessed fee of $f_R \in [0,1]$ on the collateral, we assume $x_i + \mu\sum_{j \in [n]} \left[(1-f_R)L_{ji} - L_{ij}\right] \geq 0$ for all $i\in[n]$ so that no bank defaults directly from the posting of the initial margins.
\end{assumption}

Mathematically, these decentralized clearing rules -- given the bidding strategies $\bb=(b_{ij})_{i,j\in [n]}$ -- allow for the construction of the blockchain to clear the interbank obligations.  Note that, due to the limited block capacity and no payment rehypothecation rules, multiple blocks are typically required to finish clearing the financial system.  Due to the sequential nature of the blockchain, the block number $t$ can be thought of as a (discrete) time point.
The actualized payments on a specific block $t$, following the greedy block construction in which the collected fees by the winning miner are maximized, are the bids:
\begin{align*}
\widehat \bb^t &\in \argmax \left\{ \sum_{i ,j \in [n]} \int_0^1 f b_{ij}^t(f)df \; | \; \bb^t \in \bcal_t(\widehat \bb^{1},\dots,\widehat \bb^{t-1})\right\} ,
\end{align*}
where $\bcal_t$ denotes the constraints that each block must satisfy, i.e.,
\begin{align*}
\bcal_t((\widehat \bb^s)_{s < t}) &= \left\{
\bb^t ~~\Bigg|~~
\begin{array}{l}
0\le b^t_{ij}(f) \le b_{ij}(f) - \sum_{s < t} \widehat b^s_{ij}(f) \mbox{ for all } i,j,\in[n],f\in[0,1]  ,\\
 \abs{\bigl\{(i,j)\in[n]^2 \; | \; \int_0^1 b_{ij}^t(f)df > 0\bigr\}} \leq C^\#,\\
  \sum_{j \in [n]} \int_0^1 b_{ij}^t(f)df \leq V_i^{t-1} \; \text{for all} \; {i \in [n]} \end{array}\right\} ,
\end{align*}
and $V_i^{t-1}$ denotes the cash account of bank $i$ at time $t-1$.  
That is, the bids accepted in block $t$ have not been accepted in any prior block, the block capacity is not exceeded, and the payments being placed onto the new block can be satisfied with the current cash account without rehypothecation of the simultaneous payments.\footnote{Under the relaxation of no payment rehypothecation~\eqref{np} using, e.g., UTXOs the final constraint of $\bcal_t((\widehat \bb^s)_{s < t})$ can be rewritten to allow for rehypothecation, i.e., $\sum_{j \in [n]} \int_0^1 b_{ij}^t(f)df \leq V_i^{t-1} + \sum_{j \in [n]} \int_0^1 (1-f)b_{ji}^t(f)df$.}

Importantly, as in the construction of $\bcal_t$ to determine the blocks, the cash account of each bank is needed to verify that payments can be made without rehypothecation.  With the actualized payments $\widehat \bb$ based on the greedy block construction over time, the realized cash account of bank $i$ at the completion of processing block $t$ is:
\begin{align}
\label{eq:Vt} V_i^t &= V_i^{t-1} + \sum_{j \in [n]} \int_0^1 \left[(1-f)\widehat b_{ji}^t(f) - \widehat b_{ij}^t(f)\right]df, \\
\label{eq:V0} V_i^0 &= x_i + \mu \sum_{j \in [n]} \left[(1-f_R)L_{ji} - L_{ij}\right].
\end{align}
We wish to note that, due to the constraints imposed, we have the nonnegativity of the cash account, $V_i^t \geq 0$ for every bank for all blocks $t$.

\begin{lemma}\label{lemma:unique}
There exists a clearing solution vector $(\bV^t)_{t \in \bbn}$ following~\eqref{eq:Vt} for all times $t \in \bbn := \N \cup \{0\}$.
\end{lemma}
\begin{proof}
We will approach this proof constructively.  Specifically, a process $(\bV^t)_{t \in \bbn}$ with $\bV^t=(V_1^t, \dots, V_n^t)^\T$ following~\eqref{eq:Vt} can be constructed by following Algorithm~\ref{alg:bids}.
\end{proof}
The proof of Lemma~\ref{lemma:unique} relies on the following constructive algorithm to build each block of the blockchain.  Within the optimization problem for the realized bids $\widehat \bb^t$, there may be ties over which the miners are indifferent.  As such, generally, there does not exist a unique clearing solution $(\bV^t)_{t \in \bbn}$.  Within Algorithm~\ref{alg:bids}, we consider a pro-rata assumption as in~\cite{EN01} though other rulesets could be utilized instead.\footnote{Other rulesets include, e.g., the constrained equal-awards rule, the constrained equal-losses rule, and a Talmud rule for division of assets.  We refer the interested reader to~\cite{schaarsberg2018solving, ketelaars2020decentralization}.}  This algorithm, first, finds the (as yet) unpaid liabilities; second, loops through all feasible block configurations to find that which maximizes the collected fees; third, implements that block before continuing back to step 1 until no more blocks can be constructed. 
 Within the below algorithm, for simplicity of notation, let $J(i) := \{j\in[n] \; | \; (i,j) \in J\}$ be the projection of $J \subseteq [n]^2$ on $i \in [n]$.
\begin{algorithm}\label{alg:bids}
Initialize $t = 0$ with $V_i^0 = x_i + \mu\sum_j \left[(1-f_R)L_{ji} - L_{ij}\right]$. 
\begin{enumerate}
\item\label{step2} Increment $t = t+1$.
\item\label{barb} Let $\bar b_{ij}^t := b_{ij} - \sum_{s < t} \widehat b_{ij}^s$ for every pair of banks $i,j\in[n]$.
\item\label{terminate} If $I := \{(i,j)\in[n]^2 \; | \; \int_0^1 \bar b_{ij}^t(f)df > 0 \, , \, V_i^{t-1} > 0\} = \emptyset$ then terminate with $\widehat \bb^s \equiv 0$ and $\bV^s = \bV^{t-1}$ for every $s \geq t$.
\item\label{loop} For each $J \subseteq I$ such that $\abs{J} \leq C^\#$ do: 
    \begin{enumerate}
    \item For every bank $i$ find the minimal fee level $\bar f_i^{t,J}$ that bank $i$ can support: \[\bar f_i^{t,J} = \inf\bigl\{\bar f \in [0,1] \; | \; \sum_{j \in J(i)} \int_{\bar f}^1 \bar b^t_{ij}(f)df \leq V_i^{t-1}\bigr\}.\]
    \item Define the pro-rata proportions $\pi_{ij}^{t,J}$ by:
        \begin{align*}
        \pi_{ij}^{t,J} &= \begin{cases} (1- \bar f_i^{t,J})\frac{\bar b^t_{ij}(\bar f_i^{t,J})\ind{j \in J(i)}}{\sum_{k \in J(i)} \bar b^t_{ik}(\bar f_i^{t,J})} &\text{if } \sum_{k \in J(i)} \bar b^t_{ik}(\bar f_i^{t,J}) > 0, \\ 0 &\text{if } \sum_{k \in J(i)} \bar b^t_{ik}(\bar f_i^{t,J}) = 0. \end{cases} 
        \end{align*}
    \item Construct the optimal fee structure for the miners such that only obligations between counterparties in $J$ are considered: 
    \[\widehat b_{ij}^{t,J}(f) := \begin{cases} \bar b^t_{ij}(f)\ind{f > \bar f_i^{t,J}} + \pi_{ij}^{t,J}\left[V_i^{t-1} - \sum\limits_{k \in J(i)} \int_{(\bar f_i^{t,J})^+}^1 \bar b^t_{ik}(f)df\right]\delta_{\bar f_i^{t,J}}(f) &\text{if } (i,j) \in J, \\ 0 &\text{if } (i,j) \not\in J. \end{cases}\]
    \end{enumerate}
\item\label{stepMax} Arbitrarily select $$J^t \in \argmax_{J \subseteq I} \Bigl\{\sum_{(i,j) \in J} \int_0^1 f \widehat b_{ij}^{t,J}(f)df \; | \; \abs{J} \leq C^\#\Bigr\}.$$ 
\item\label{construct} Construct the realized bids $\widehat \bb^t := \widehat \bb^{t,J^t}$ 
 and update the cash account of each bank $i$ by $$V_i^t = V_i^{t-1} + \sum_{j \in [n]} \int_0^1 \left[(1-f)\widehat b_{ji}^t(f) - \widehat b_{ij}^t(f)\right]df.$$
\item Return to Step~\eqref{step2}.
\end{enumerate}
\end{algorithm}

For each block $t$, Algorithm~\ref{alg:bids} first computes at step~\eqref{barb}, for every pair of banks $i,j\in[n]$, the bids $\bar b_{ij}^t$ remaining from the total $b_{ij}$ for obligation of bank $i$ to $j$.  That is, $\bar \bb^t$ are those bids which have not already cleared prior to block $t$ at every fee level $f\in[0,1]$. At step~\eqref{terminate}, this algorithm finds the set $I$ of all pair of banks $(i,j)\in [n]^2$ with a positive uncleared obligation and for which the obligor has a positive cash account. If this set is empty, then obviously the algorithm terminates as no more payments can be made. At step~\eqref{loop}, for each potential subset (i.e., subgraph) $J\subseteq I$ with $\abs{J} \leq C^\#$, this algorithm computes, the minimal fee level $\bar f_i^{t,J}$ that bank $i$ can support and the resulting bids that would be realized provided only obligations between counterparties in this subgraph are considered.  (If discrete bids are permitted, the pro-rata proportions $\pi_{ij}^{t,J}$ for the remaining bids $\bar b_{ij}^t$ cleared at the fee level $\bar f_i^{t,J}$ is also constructed). At step~\eqref{stepMax}, the algorithm optimally selects the subgraph $J$ which maximizes the fees collected by the winning miner. Finally, at step~\eqref{construct}, the algorithm implements block $t$ and updates the cash account of each bank.

We wish to conclude our discussion of blockchain clearing systems by considering the limiting behavior of the cash account $\bV^t$ as the block number $t$ tends to infinity.
\begin{proposition}\label{prop:limit}
Let $(\bV^t)_{t\in\bbn}$ be constructed from Algorithm~\ref{alg:bids}.  Then the limit $\lim_{t \to \infty} V_i^t$ exists for every bank $i$.
\end{proposition}
\begin{proof}
First, consider the sequence of remaining bids $(\bar b_{ij}^t)_{t \in \bbn} \subseteq L^1([0,1])$ for any pair of banks $i,j$ from Algorithm~\ref{alg:bids}.  By construction $0 \leq \bar b_{ij}^t(f) \leq \bar b_{ij}^{t-1}(f)$ for any $f \in [0,1]$.  Therefore, by monotone convergence there exists a pointwise limit $\bar b_{ij}^* := \lim_{t \to \infty} \bar b_{ij}^t$ for any pair of banks $i,j$.
We will now prove the convergence of $V_i^t$ by noting that $$V_i^t = V_i^0 + \sum_{j \in [n]} \int_0^1 \left[(1-f)[b_{ji}(f) - \bar b^t_{ji}(f)] - [b_{ij}(f) - \bar b^t_{ij}(f)]\right]df$$ for any time $t$ and monotone convergence implies \[\lim_{t \to \infty} V_i^t = V_i^0 + \sum_{j \in [n]} \int_0^1 \left[(1-f)[b_{ji}(f) - \bar b^*_{ji}(f)] - [b_{ij}(f) - \bar b^*_{ij}(f)]\right]df.\]  
\end{proof}

\subsection{Limiting Net Worths}\label{sec:limit}

As considered implicitly with Proposition~\ref{prop:limit}, the blockchain clearing system can result in an infinite number of blocks being required to clear the financial system.  This infinite number of block notion, in fact, lends itself to a simple clearing fixed point akin to a prioritized Eisenberg-Noe system as presented in, e.g.,~\cite{E07} but with a cost structure for the seniority levels of debt when the block capacity $C^\#$ is large enough.  Let $K_i^*$ denote the \emph{terminal} net worth of bank $i$ after the (possibly infinite length) clearing procedure as presented above; this interpretation is made explicit in Proposition~\ref{prop:limit-K} below.  Notably, the results presented within this section do not rely on~\eqref{np} no payment rehypothecation as the limit in (block) time allows for all payments to eventually be reused as necessary.

These terminal net worths $\bK^*=(K^*_1,\dots, K^*_n)^\T$ are constructed so that bank $j$ covers all of its liabilities requested at a fee level above $f_j^*$ (and proportionally at $f_j^*$); this level $f_j^*$ is the minimal fee that bank $j$ is able to satisfy at least partially with its accumulated assets.  The construction of a single threshold fee level for each bank can be taken because $C^\#$ is sufficiently large which guarantees the miners will always select the highest fees without consideration for how many counterparties are involved. With this notion, the payments made by bank $j$ depend on whether $j$ is solvent ($K_j^* \geq 0$) or defaulting ($K_j^* < 0$).  
If bank $j$ is solvent then it must cover all of its obligations in full and bank $i$ receives $\int_0^1 (1-f)b_{ji}(f)df$.  
If bank $j$ defaults on some of its obligations then, as mentioned previously, first it covers only its obligations strictly above $f_j^*$ so that bank $i$ receives \[\int_{(f_j^*)^+}^1 (1-f)b_{ji}(f)df = \lim\limits_{\epsilon\searrow0} \int_{(f_j^*)+\epsilon}^1 (1-f)b_{ji}(f)df.\] 
However, if there is a positive measure of mass of obligations at $f_j^*$ (i.e., $\int_{f_j^*}^{(f_j^*)^+} b_{ji}(f)df > 0$) then $j$ may have a cash surplus with which it covers part of this mass of obligations.  This cash surplus is exactly equal to the difference between the assets held by $j$ (i.e., $K_j^* + \sum_{k \in [n]} L_{jk}$) and the obligations with fee levels above $f_j^*$ (i.e., $\sum_{k \in [n]} [\mu L_{jk} + \int_{(f_j^*)^+}^1 b_{jk}(f)df]$).   The proportional repayment scheme described in Algorithm~\ref{alg:bids} implies this cash surplus is distributed following the rates $\pi_{ji}^*$ defined in~\eqref{eq:pi} below with an explicit cut provided for the miners. 
That is, $\bK^*$ are these terminal net worths if for every bank $i \in [n]$ they satisfy the fixed point equations
\begin{align}
\label{eq:K} K_i^* &= x_i - \sum_{j \in [n]} L_{ij}+ \sum_{j \in [n]} \left[\mu(1-f_R)L_{ji} + \ind{K_j^* \geq 0} \int_0^1 (1-f)b_{ji}(f)df\right.\\
\nonumber    &\quad \left. + \ind{K_j^* < 0}\left(\int_{(f_j^*)^+}^1 (1-f)b_{ji}(f)df + \pi_{ji}^*\left(K_j^* + \sum_{k \in [n]} \left[(1-\mu)L_{jk} - \int_{(f_j^*)^+}^1 b_{jk}(f)df\right]\right)\right)\right], 
\end{align}
where
\begin{align}
\label{eq:pi} \pi_{ij}^* &= \begin{cases} (1-f_i^*)\frac{b_{ij}(f_i^*)}{\sum_{k \in [n]} b_{ik}(f_i^*)} &\text{if } \sum_{k \in [n]} b_{ik}(f_i^*) > 0 \\ 0 &\text{if } \sum_{k \in [n]} b_{ik}(f_i^*) = 0, \end{cases} \quad \text{for all} \quad j \in [n],\\
\label{eq:f} f_i^* &= \inf\Bigl\{f_i \in [0,1] \; | \; K_i^* + \sum_{j \in [n]} L_{ij} \geq \sum_{j \in [n]} \Bigl[\mu L_{ij} + \int_{f_i^+}^1 b_{ij}(f)df\Bigr]\Bigr\}.
\end{align}
The key element of this clearing net worths problem are the levels $\bdf^*=(f^*_1, \dots, f^*_n)^\T \in [0,1]^n$ denoting the minimal fee (i.e., the lowest seniority) that each bank is able to satisfy in full.  Bank $i$ may be able to satisfy a fraction of its obligations at the seniority $f_i^*$ (at which level we assume a pro-rata repayment scheme as in Algorithm~\ref{alg:bids} so as to coincide with the Eisenberg-Noe framework), but no payments are made on any bids $b_{ij}(f)$ for $f < f_i^*$.
Note that, for any feasible wealth level $K_i$ of bank $i$, we have 
$$K_i + \sum_{j \in [n]} \int_0^1 b_{ij}(f)df = K_i + (1-\mu)\sum_{j \in [n]} L_{ij} \geq x_i + \mu \sum_{j \in [n]} \left[(1-f_R) L_{ji} - L_{ij}\right] \geq 0$$ by Assumption~\ref{assump:initial}; as such, the optimization problem defining the defaulting fee level $f_i^*$ of bank $i$ is always feasible at $f_i = 1$.

\begin{lemma}\label{lemma:K}
Consider the blockchain system with sufficiently large capacity, i.e., \\
$C^\# \geq \abs{\{(i,j)\in [n]^2 \; | \; L_{ij} > 0\}}.$  There exists a greatest and least clearing solution $\bK^\uparrow \geq \bK^\downarrow$ to~\eqref{eq:K}.  Moreover, if $x_i + \mu(1-f_R)\sum_{j \in [n]} L_{ji} \geq 0$ for all $i \in [n]$, then there exists a unique clearing solution $\bK^\uparrow = \bK^\downarrow$. 
\end{lemma}
\begin{proof}
Before proving the existence and uniqueness of the clearing solution, we wish to note that the key components $\bK^*$ and $\bold{f}^*$ defined in~\eqref{eq:K} and~\eqref{eq:f} can be simplified in equilibrium as:
\begin{align*}
K_i^* &= x_i - \sum_{j \in [n]} L_{ij}+ \sum_{j \in [n]}\left[\mu(1-f_R)L_{ji} + \ind{K_j^* \geq 0} \int_0^1 (1-f)b_{ji}(f)df\right.\\ 
    &\quad \left. + \ind{K_j^* < 0}\left(\int_{(f_j^*)^+}^1 (1-f)b_{ji}(f)df + \pi_{ji}^*\left(K_j^* + \sum_{k \in [n]} \int_0^{f_j^*} b_{jk}(f)df\right)\right)\right] ,\\
f_i^* &= \inf\left\{f_i \in [0,1] \; | \; K_i^* + \sum_{j \in [n]} \int_0^{f_i} b_{ij}(f)df \geq 0\right\}.
\end{align*}

The existence of a greatest and least clearing net worth $\bK^\uparrow \geq \bK^\downarrow$ follows from an application of Tarski's fixed point theorem as $f_i^*$ is nonincreasing in $K_i^*$ and the resulting monotonicity of the clearing equation.

We will prove uniqueness by following the same general strategy as that used in~\cite{EN01}.  Specifically, we will first prove that the positive equities are equivalent under any clearing solution; we will then use that result to reduce the problem to an equivalent Eisenberg-Noe system under which uniqueness results already exist.
To consider the positive equities, first note that $(K_i^\uparrow)^+ \geq (K_i^\downarrow)^+$ for any bank $i$.  Now let $\bK^*$ denote an arbitrary clearing solution, then the positive equity for bank $i$ is provided by
\begin{align*}
(K_i^*)^+ &= x_i + \sum_{j \in [n]} \left[\mu(1-f_R)L_{ji} + \ind{K_j^* \geq 0}\int_0^1 (1-f)b_{ji}(f)df\right.\\ 
    &\quad \left. + \ind{K_j^* < 0}\left(\int_{(f_j^*)^+}^1 (1-f)b_{ji}(f)df + \pi_{ji}^*\left[K_j^{*} + \sum_{k \in [n]} \int_0^{f_j^*} b_{jk}(f)df\right]\right)\right] \\
    &\quad - \left[\sum_{j \in [n]} L_{ij} - (K_i^{*})^-\right]\\
&= x_i + \sum_{j \in [n]} \left[\mu(1-f_R)L_{ji} + \ind{K_j^* \geq 0}\int_0^1 (1-f)b_{ji}(f)df \right.\\ 
    &\quad \left. + \ind{K_j^* < 0}\left(\int_{(f_j^*)^+}^1 (1-f)b_{ji}(f)df + \pi_{ji}^*\left[K_j^{*} + \sum_{k \in [n]} \int_0^{f_j^*} b_{jk}(f)df\right]\right)\right]\\ 
    &\quad - \left[\sum_{j \in [n]} \mu L_{ij} + \ind{K_i^* \geq 0} \sum_{j \in [n]} \int_0^1 b_{ij}(f)df\right]\\
        &\quad -\ind{K_i^* < 0} \left[ \sum_{j \in [n]} \int_{(f_i^*)^+}^1 b_{ij}(f)df + K_i^* + \sum_{j \in [n]} \int_0^{f_i^*} b_{ij}(f)df\right].
\end{align*}
As a direct consequence, the total amount of positive equity in the system is given by:
\begin{align*}
&\sum_{i \in [n]} (K_i^*)^+ = \sum_{i \in [n]} x_i + \mu(1-f_R)\sum_{i \in [n]}\sum_{j \in [n]} L_{ij} + \sum_{i \in [n]}\sum_{j \in [n]}\ind{K_i^* \geq 0}\int_0^1 (1-f)b_{ij}(f)df\\ 
    &\quad + \sum_{i \in [n]}\sum_{j \in [n]}\ind{K_i^* < 0}\int_{(f_i^*)^+}^1 (1-f)b_{ij}(f)df + \sum_{i \in [n]}\ind{K_i^* < 0}(1-f_i^*)\left[K_i^* + \sum_{j \in [n]}\int_0^{f_i^*}b_{ij}(f)df\right]\\ 
    &\quad - \mu\sum_{i \in [n]}\sum_{j \in [n]} L_{ij} - \sum_{i \in [n]}\sum_{j \in [n]}\ind{K_i^* \geq 0}\int_0^1 b_{ij}(f)df - \sum_{i \in [n]}\sum_{j \in [n]}\ind{K_i^* < 0} \int_{(f_i^*)^+}^1 b_{ij}(f)df\\ 
    &\quad - \sum_{i \in [n]} \ind{K_i^* < 0}\left[K_i^* + \sum_{j \in [n]} \int_0^{f_i^*} b_{ij}(f)df\right]\\
&= \sum_{i \in [n]} x_i - \mu f_R \sum_{i \in [n]}\sum_{j \in [n]} L_{ij} - \sum_{i \in [n]}\ind{K_i^* \geq 0}\sum_{j \in [n]}\int_0^1 fb_{ij}(f)df\\ 
    &\quad - \sum_{i \in [n]}\ind{K_i^* < 0}\sum_{j \in [n]}\int_{(f_j^*)^+}^1 fb_{ij}(f)df - \sum_{i \in [n]}\ind{K_i^* < 0}f_i^*\left[K_i^* + \sum_{j \in [n]}\int_0^{f_i^*}b_{ij}(f)df\right]\\
&= \sum_{i \in [n]} x_i - \mu f_R \sum_{i \in [n]}\sum_{j \in [n]} L_{ij} - \sum_{i \in [n]}\underbrace{\left[\sum_{j \in [n]} \int_{(f_i^*)^+}^1 f b_{ij}(f)df + f_i^*\left(K_i^* + \sum_{j \in [n]} \int_0^{f_i^*} b_{ij}(f)df\right)\right]}_{=: G_i(K_i^*)}.
\end{align*}
The final equation above follows since $K_i^* \geq 0$ implies $f_i^* = 0$ by construction of the minimal fee level $f_i^*$.  
As $\sum_{i \in [n]} x_i - \mu f_R \sum_{i \in [n]}\sum_{j \in [n]} L_{ij}$ is a constant, $\sum_{i \in [n]} (K_i^\uparrow)^+ > \sum_{i \in [n]} (K_i^\downarrow)^+$ if and only if $\sum_{i \in [n]} G_i(K_i^\uparrow) < \sum_{i \in [n]} G_i(K_i^\downarrow)$.  However, by monotonicity of the minimal fee level $f_i^\uparrow \leq f_i^\downarrow$, it trivially holds that $G_i(K_i^\uparrow) \geq G_i(K_i^\downarrow)$ and thus $\sum_{i \in [n]} (K_i^\uparrow)^+ \leq \sum_{i \in [n]} (K_i^\downarrow)^+$.  
As a direct consequence, it must hold that $(K_i^\uparrow)^+ = (K_i^\downarrow)^+$ for every bank $i$.

In fact, by observation, $G_i(K_i^\uparrow) = G_i(K_i^\downarrow)$ only if $f_i^\uparrow = f_i^\downarrow = 0$ or $K_i^\uparrow = K_i^\downarrow$. Assume now that $K_i^\uparrow > K_i^\downarrow$ for some bank $i$.  By the equality of the positive equities, it must therefore hold that $0 \geq K_i^\uparrow > K_i^\downarrow$.  As the equality of the system-wide positive equity implies that $\sum_{j \in [n]} G_j(K_j^\uparrow) = \sum_{j \in [n]} G_j(K_j^\downarrow)$, it must hold that $f_i^\uparrow = f_i^\downarrow = 0$ for any bank $i$ such that $K_i^\uparrow > K_i^\downarrow$.  Finally, as all other banks have unique clearing net worths, any clearing net worths $\bK^*$ must be the result of an Eisenberg-Noe clearing system defined by:
\begin{align*}
&\tilde x_i := x_i \\
&+ \sum_{j \in [n]}\left[\mu(1-f_R)L_{ji} + \int_{(f_j^\uparrow)^+}^1 (1-f)b_{ji}(f)df + \ind{K_j^\uparrow < 0}\ind{f_j^\uparrow > 0}\pi_{ji}^\uparrow\left(K_j^\uparrow + \sum_{k \in [n]}\int_0^{f_j^\uparrow} b_{jk}(f)df\right)\right],\\
&\tilde L_{ij} := \ind{f_i^\uparrow = 0}\int_0^{0^+}b_{ij}(f)df.
\end{align*}
By assumption this is a regular financial network\footnote{A financial system is regular if every risk orbit $o(i) = \{j \in [n] \; | \; \text{there exists a directed path from $i$ to $j$}\}$ is a surplus set ($\tilde L_{jk} = 0$ for every $(j,k) \in o(i) \times o(i)^c$ and $\sum_{j \in o(i)} \tilde x_j > 0$); see \cite[Definition 5]{EN01}.} and thus the greatest and least clearing vectors are the same and uniqueness holds by~\cite[Theorem 2]{EN01}.
\end{proof}

\begin{remark}
The clearing procedure with $C^\# \geq \abs{\{(i,j)\in[n]^2 \; | \; L_{ij} > 0\}}$ can readily be compared with the prioritized Eisenberg-Noe system over a (potential) \emph{continuum} of seniority levels.  The distinction between these models lies fully in the subtraction of the fees from interbank assets within the construction within the decentralized system.  These fees can therefore be thought of as the explicit cost of decentralization and provide the dead-weight losses caused by the miners.  (Though a centralized system may have fees or other onerous regulations, we do not model those herein and are only considering the same parameters applied to the centralized and decentralized clearing systems.)
\end{remark}

We now wish to formally relate the terminal net worths $\bK^*$ to the blockchain clearing system as presented in Algorithm~\ref{alg:bids}.  Specifically, the limiting cash account $\lim_{t\to\infty} V_i^t$ is the positive equity $(K_i^*)^+$ of bank $i$.
\begin{proposition}\label{prop:limit-K}
Consider the blockchain system with sufficiently large capacity, i.e., \\ $C^\# \geq \abs{\{(i,j)\in[n]^2 \; | \; L_{ij} > 0\}}$. Under the uniqueness condition of Lemma~\ref{lemma:K}, $(K_i^*)^+ = \lim_{t \to \infty} V_i^t$ for all $i\in [n]$, where $(\bV^t)_{t \in \bbn}$ is constructed as in Algorithm~\ref{alg:bids}.
\end{proposition}
\begin{proof}
Within Algorithm~\ref{alg:bids} with $C^\# \geq \abs{\{(i,j)\in[n]^2 \; | \; L_{ij} > 0\}}$, the loop over $J$ need not be considered as $J = I$ can be taken without loss of generality.  As such, for $i\in[n]$, we can consider explicitly $\bar f_i^t$ without the need for the superscript $J$.  Note that $\bar f_i^t \in [0,1]$ is nonincreasing over blocks $t$ for every bank $i$.  Therefore, there exists some limit $\bar f_i^* = \lim_{t \to \infty} \bar f_i^t$. Furthermore, by Proposition~\ref{prop:limit}, $\bV^* = \lim_{t \to \infty} \bV^t$ exists as well.  In fact, by construction within Algorithm~\ref{alg:bids}, $\bV^*$ solves the following fixed point problem: 
\begin{align}
\nonumber V_i^* &= \left(\begin{array}{l}x_i + \sum_{j \in [n]}\left[\mu(1-f_R)L_{ji} + \ind{V_j^* > 0} \int_0^1 (1-f)b_{ji}(f)df\right.\\ \quad\left.+ \ind{V_j^* = 0}\left(\int_{(\bar f_j^*)^+}^1 (1-f)b_{ji}(f)df + \bar\pi_{ji}^* S_j(\bar {\bdf}^*)\right)\right] - \sum_{j \in [n]} L_{ij}\end{array}\right)^+, \quad \forall i \in [n],
\end{align}
where, for every $i,j \in [n]$,
\begin{align}
\label{eq:S} S_i(\bar {\bdf}^*) &= x_i + \sum_{j \in [n]}\left[\mu(1-f_R)L_{ji} + \ind{V_j^* > 0} \int_0^1(1-f)b_{ji}(f)df\right.\\ 
\nonumber    &\quad \left.+ \ind{V_j^* = 0}\left(\int_{(\bar f_j^*)^+}^1(1-f)b_{ji}(f)df + \bar\pi_{ji}^* S_j(\bar {\bdf}^*)\right)\right] - \sum_{j \in [n]} \left[\mu L_{ij} + \int_{(\bar f_i^*)^+}^1 b_{ij}(f)df\right],\\
\nonumber \bar\pi_{ij} &= \begin{cases} (1- \bar f_i^*)\frac{b_{ij}(\bar f_i^*)}{\sum_{k \in [n]} b_{ik}(\bar f_i^*)} &\text{if } \sum_{k \in [n]} b_{ik}(\bar f_i^*) > 0, \\ 0 &\text{if } \sum_{k \in [n]} b_{ik}(\bar f_i^*) = 0, \end{cases} 
\end{align}
and $\bar {\bdf}^*=(\bar f_1^*, \dots, \bar f_n^*)^\T$.  In this fixed point problem $S_i(\bar {\bdf}^*)$ denotes the surplus assets after covering all bids strictly above the threshold fee $\bar f_i^*$.  
Note that $S_j(\bar {\bdf}^*)$ only appears if $V_j^* = 0$ and this surplus is bounded from below by $0$ and from above by $\sum_{j \in [n]} \int_{\bar f_i^*}^{(\bar f_i^*)^+} b_{ij}(f)df$.
Now, define for every bank $i$
$$K_i^* := V_i^* \ind{V_i^* > 0} + \Big(S_i(\bar {\bdf}^*) - \sum_{j \in [n]} \int_0^{\bar f_i^*} b_{ij}(f)df\Big)\ind{V_i^* = 0},$$ so that $\bK^*$ solves~\eqref{eq:K} with threshold fees $\bar {\bdf}^*$ by construction.  Therefore, by uniqueness of $\bK^*$ via Lemma~\ref{lemma:K}, the result holds if $$\bar f_i^* = \inf\bigl\{f_i \in [0,1] \; | \; K_i^* + \sum_{j \in [n]} \int_0^{f_i} b_{ij}(f)df \geq 0\bigr\} =: f_i^*$$ for every bank $i$.  
First, consider the case in which $K_i^* > 0$ then $K_i^* = V_i^*$ and $f_i^* = 0$. If $\bar f_i^* > f_i^*$ then Algorithm~\ref{alg:bids} can continue as there are more obligations to be fulfilled that can be covered; this contradicts the fact that $\bar f_i^* > f_i^*$, as a result, $\bar f_i^* = 0$ as well.
Second, consider the case in which $K_i^* \leq 0$ then $K_i^* = S_i(\bar {\bdf}^*) - \sum_{j \in [n]} \int_0^{\bar f_i^*} b_{ij}(f)df$ and $f_i^* = \inf\{f_i \in [0,1] \; | \; S_i(f_i) \geq 0\}$; if $f_i^* < \bar f_i^*$ then $S_i(\bar {\bdf}^*) > \sum_{j \in [n]} \int_{\bar f_i^*}^{(\bar f_i^*)^+} b_{ij}(f)df$ and if $f_i^* > \bar f_i^*$ then $S_i(\bar {\bdf}^*) < 0$, either of which contradict its bound as the surplus assets.
\end{proof}

We wish to conclude this section by noting that the equilibrium condition for the terminal net worths $\bK^*$ as provided in~\eqref{eq:K} simplifies significantly if all bids $b_{ij}$, for every pair of banks $i,j$, are continuous so that the surplus (considered explicitly in~\eqref{eq:S} in the above proof) for defaulting firms is always equal to $0$.  Specifically, under this setting the limiting net worths can be defined as:
\begin{align*}
K_i^* &= x_i + \sum_{j \in [n]} \left[\mu(1-f_R)L_{ji} + \int_{f_j^*}^1 (1-f)b_{ji}(f)df\right] - \sum_{j \in [n]} L_{ij},\\
f_i^* &= \inf\left\{f_i \in [0,1] \; | \; K_i^* + \sum_{j \in [n]} \int_0^{f_i} b_{ij}(f)df \geq 0\right\}.
\end{align*}
Notably, the same uniqueness proof given in Lemma~\ref{lemma:K} can be used to prove the uniqueness in this setting without any assumption on the assets of each bank.

\section{Optimal Bidding}\label{sec:bidding}


In the above construction of decentralized clearing, the bidding strategies $b_{ij}: [0,1] \to \bbr_+ \cup \{\infty\}$ for any pair of banks $i,j$ have been assumed to be known.  We now want to bring our attention to optimal constructions for such bidding strategies.  As will be seen, the optimal bidding strategy for one obligation may depend on the bids submitted for all other obligations even for distinct counterparties.  We consider two distinct approaches to forming the optimal bidding strategies: \begin{enumerate*} \item as the result of a complex non-cooperative game; and \item as a Pareto optimal strategy.\end{enumerate*} 

Herein, due to the natural discretization that occurs from, e.g., the limited divisibility of assets in practice, we limit ourselves to \emph{discrete} bidding strategies.  Specifically, bids on the (unsecured portion of the) obligation $(1-\mu)L_{ij}$ can only be made in amounts 
$$\dcal_{ij}(D) := [0,(1-\mu)L_{ij}]\cap\left(\frac{1}{D}\bbn\right),$$ 
(e.g., in dollars and cents or bitcoins and satoshis) at fee levels 
$$\fcal(F) := [0,1]\cap\left(\frac{1}{F}\bbn\right),$$ 
for discretization parameters $D,F\in\N$.\footnote{Other finite discretizations of the bids and fees can be taken without altering the following results.  For instance, bids can be placed in percentage terms instead, i.e., $\dcal_{ij}(D) := (1-\mu)L_{ij} \times \left([0,1]\cap\left(\frac{1}{D}\bbn\right)\right)$.}
\begin{assumption}\label{ass:discrete}
For the remainder of this work, we will assume that the discretization of bids $D\in\N$ satisfies $D(1-\mu)L_{ij} \in\bbn$ for all (unsecured) obligation $(1-\mu)L_{ij}$ as any partial obligation would never be able to be requested in a bid.
\end{assumption}
Therefore, a feasible bidding strategy $b_{ij}: [0,1] \to \bbr_+ \cup \{\infty\}$ for (unsecured) obligation $(1-\mu)L_{ij}$ must be of the form $$b_{ij}(\cdot) := \sum_{f \in \fcal(F)} d_{ij}(f) \delta_f(\cdot)$$ for Dirac delta function $\delta$ and set of bid amounts $d_{ij}: \fcal(F) \to \dcal_{ij}(D)$ such that $$\sum_{f \in \fcal(F)} d_{ij}(f) = (1-\mu)L_{ij}.$$
We will denote the set of all feasible bidding strategies $b_{ij}$ on the (unsecured) obligation $(1-\mu)L_{ij}$ with discretizations $D,F$ by $\bcal_{ij}(D,F)$.

\subsection{Nash Equilibrium Bidding}\label{sec:bidding-nash}
Within this section we wish to consider the setting in which the banks choose their bidding strategy as the result of a complex non-cooperative game.  In such a way, the bids placed for this blockchain clearing system must be a Nash equilibrium.  

Formally, each bank $j$ must submit, within the smart contracts encoding its incoming (unsecured) obligations $(1-\mu)L_{ij}$ from all banks $i$, the bidding strategy $b_{ij} \in \bcal_{ij}(D,F)$. Such bidding strategies can be either pure or mixed, i.e., deterministic or random.  Conceptually, though a bank may prefer the certainty of a deterministic bidding strategy so that the costs are known at the implementation of the contract, the act of encoding a known fee schedule within the smart contract (viewable to all participants of the blockchain) subjects the bank to the possibility of predation.  This predation can take the form of another bank simply outbidding the known fees on the same counterparty (who is at risk of default) and, therefore, capturing the payments made before that counterparty defaults.  As such, the optimal bidding strategy for any \emph{risky} obligation may incorporate an element of randomness, i.e., a mixed strategy.  In fact, enforcing deterministic bids can lead to the non-existence of an equilibrium optimal bidding strategy as is detailed in Example~\ref{ex:no-pure} below.  These mixed strategies can be written into a smart contract with, e.g., a random number generator that is called at the execution time of the contract.  Mathematically, we will abuse notation and define a mixed strategy as a random function $b_{ij}: \Omega \to \bcal_{ij}(D,F)$ for all pairs of banks $i,j\in [n]$ and for some probability space $(\Omega,\bbf,\bbp)$. For convenience we will suppress the dependency on $\om$ and simply write $b_{ij}(f) = (b_{ij}(\om))(f).$ We will denote the set of mixed strategies for the obligation $(1-\mu)L_{ij}$ from bank $i$ to $j$ by $L^{0}[\bcal_{ij}(D,F)] := \{b_{ij}: \Omega \to \bcal_{ij}(D,F)\}$.\footnote{Due to the discrete setting considered herein, measurability of $b_{ij} \in L^{0}[\bcal_{ij}(D,F)]$ follows trivially.}  As detailed in Corollary~\ref{cor:nash-exist}, there will always exist some optimal mixed bidding strategy.

In order to construct a bidding strategy -- either pure or mixed -- each bank will act as a utility maximizer within the financial system over its own (terminal) clearing equity as computed through Algorithm~\ref{alg:bids}.  (Recall that the limiting equity exists due to Proposition~\ref{prop:limit}.) That is, the utility gained by bank $i$ given a complete set of bidding strategies $(b_{ij}^\dagger)_{i,j}$ -- either pure or mixed -- is the expected utility
\[\lim_{t\to\infty}\E[u_i(V_i^t(\bb^\dagger))]\]
where the utility function $u_i$ is only applied to the positive equity (i.e., the cash account) as a defaulting institution does not alter its viewpoint based on the scale of the default.  
The expectation is taken over three distinct (and independent notions): \begin{enumerate*} \item the randomness of a mixed bidding strategy; \item a possible random stress scenario as is considered in, e.g.,~\cite{AOT15}; and \item the randomness introduced by taking an arbitrary maximizer in Algorithm~\ref{alg:bids}\eqref{stepMax}.\end{enumerate*}
Note, by Proposition~\ref{prop:limit-K}, in the $C^\# \geq \abs{\{(i,j)\in [n]^2\mid L_{ij} > 0\}}$ setting of Section~\ref{sec:limit}, the expected utility is taken over $[K_i^*(\bb^\dagger)]^+$.
Therefore, bank $i$ seeks to optimize its choice of bids ($b_{ji}$ for obligations $(1-\mu)L_{ji}$) given the bidding strategies taken by all other institutions, i.e.,
\begin{align}
\label{eq:nash}
\bb_{\cdot i}^\dagger &\in \argmax\left\{\lim_{t\to\infty}\E[u_i(V_i^t(\bb_{\cdot i},(\bb_{\cdot k}^\dagger)_{k\neq i}))] \; | \; b_{ji} \in L^{0}[\bcal_{ji}(D,F)] \; \text{ for all} \; j\in [n] \right\}.
\end{align}
If a pure strategy is desired by some bank $i$, then bank $i$ merely updates its strategy space in~\eqref{eq:nash} to be the deterministic strategies $\bcal_{ji}(D,F)$ for every counterparty $j$.

\begin{corollary}\label{cor:nash-exist}
There exists a mixed strategy Nash equilibrium bidding strategy to~\eqref{eq:nash} for any blockchain clearing system with finite discretization $D,F \in \N$.
\end{corollary}
\begin{proof}
This follows directly from Nash's theorem~\cite{nash1950} due to the discretization of the fee and bidding levels.
\end{proof}

\begin{remark}\label{rem:no-competition}
If $x_i + \sum_{j\in[n]} [\mu(1-f_R) + (1-\mu)]L_{ji} - \sum_{j\in[n]} L_{ij} \geq 0$ for every bank $i \in [n]$, i.e., there are no defaults in the system if all obligations are paid under zero fees, then the universal application of this zero fee -- a pure strategy -- is a Nash equilibrium bidding strategy.  Furthermore this is a Pareto optimal bidding strategy for the system of banks as the maximum possible amount of capital is retained by the banks.  (Pareto optimal bidding strategies are formally presented in Section~\ref{sec:bidding-pareto} below.)  

Intriguingly, this leads to the conclusion that blockchain miners increase the collected transaction fees if there is some default risk towards these interbank contracts. Such default risk increases the fees paid only up to a point though as the total payments drop as banks are in worse financial health which can counteract the higher fees.
\end{remark}

Though Corollary~\ref{cor:nash-exist} guarantees the existence of a mixed optimal bidding strategy, we now wish to formally demonstrate that a pure strategy equilibrium need not exist.  We will do this in a simple, three bank, system in which the implications of all pure strategies are easily interpretable. 
\begin{example}\label{ex:no-pure}
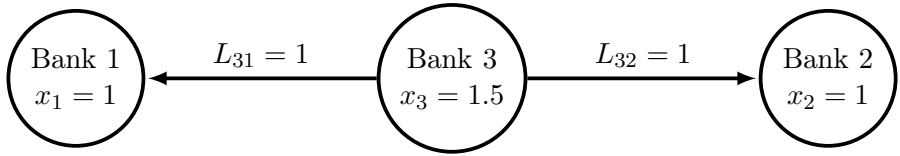
\begin{figure}[h!]
\centering
\begin{tikzpicture}
\tikzset{node style/.style={state, minimum width=0.36in, line width=0.5mm, align=center}}
\node[node style] at (0,0) (x1) {Bank $1$\\$x_1 = 1$};
\node[node style] at (10,0) (x2) {Bank $2$\\$x_2 = 1$};
\node[node style] at (5,0) (x3) {Bank $3$\\$x_3 = 1.5$};

\draw[every loop, auto=right, line width=0.5mm, >=latex]
(x3) edge node {$L_{31} = 1$} (x1)
(x3) edge node[above] {$L_{32} = 1$} (x2);
\end{tikzpicture}
\caption{A simple network topology for Example~\ref{ex:no-pure}.}
\label{fig:no-pure-network}
\end{figure}
Consider the three bank uncollateralized ($\mu = 0$) system displayed in Figure~\ref{fig:no-pure-network}, i.e., in which bank $3$ owes $\$1$ to both bank $1$ and bank $2$ but has only $\$1.50$ in liquid assets and no other obligations exist within this system.  To simplify the discussion, we will assume $C^\# \geq 2$ herein.  The bids need to be submitted in full dollar units, i.e., only integer values are accepted ($D = 1$), and all the fees must be a multiple of $10\%$ ($F = 10$).  As such, and due to the network under consideration, the feasible bids are of the form $b_{3i}(f) = \delta_{f_i}(f)$ for $f_i \in \fcal(10)$ for both banks $i = 1,2$.  That is, we can completely characterize bank $i$'s feasible bidding strategies $\bcal_{3i}(1,10) = \{\delta_{f_i}(\cdot) \; | \; f_i \in \fcal(10)\}$ by the choice $f_i \in \fcal(10)$ without the need to consider the size of the bid.

Herein we will assume both banks $1$ and $2$ are risk-neutral and solely wish to maximize the (expected) payment received from bank $3$. Due to the construction of the interbank system and the eligible bids, the possible payments to bank $1$ and $2$ can easily be found.  Specifically, the payment received by bank $1$ is given by the mapping $P_1: \fcal(10)^2 \to \bbr_+$ defined by:
\begin{align*}
P_1(f_1,f_2) &= \left(1 - f_1\right) \times \begin{cases} \$1.00 &\text{if } f_1 > f_2 ,\\ \$0.75 &\text{if } f_1 = f_2, \\ \$0.50 &\text{if } f_1 < f_2, \end{cases}
\end{align*}
where bank $1$ bids at the (fixed) fee of $f_1 \in \fcal(10)$ and bank $2$ at $f_2 \in \fcal(10)$.  
As there are no other obligations within this financial system, $P_1$ completely characterizes bank $1$'s objective.
By symmetry of the banks, bank $2$'s payoff is given by the mapping $P_2(f_1,f_2) = P_1(f_2,f_1)$ which it is attempting to maximize.

First, we consider the possibility for pure strategies.  We do this by considering the best response function $f_1^\dagger: \fcal(10) \to 2^{\fcal(10)}$ for player $1$ given player $2$'s bidding fee level.  That is, $f_1^\dagger(f_2)$ provides the fee that bank $1$ would select so as to maximize its expected payout $P_1$ given that bank $2$ is bidding at fee level $f_2 \in \fcal(10)$.  
(Due to the symmetry between banks, the best response function $f_2^\dagger: \fcal(10) \to 2^{\fcal(10)}$ for bank $2$ is identical to that for bank $1$, i.e., $f_2^\dagger(f) = f_1^\dagger(f)$ for any $f \in \fcal(10)$.)
This best response function is computed as:
\begin{align*}
f_1^\dagger(f_2) &= \begin{cases} f_2 + 0.1 &\text{if } f_2 < 0.4, \\ \{0 \, , \, 0.5\} &\text{if } f_2 = 0.4, \\ 0 &\text{if } f_2 > 0.4, \end{cases} \quad \; \text{for all} \; f_2 \in \fcal(10).
\end{align*}
This best response function can be understood by considering the three cases:
\begin{enumerate}
\item If $f_2 < 0.4$: bank $1$ has three fees to consider \begin{enumerate*} \item exceeding the fee level of bank $2$ by as little as possible ($f_1 = f_2 + 0.1$) to recover $\$(0.9-f_2)$, \item matching bank $2$ ($f_1 = f_2$) to recover $\$0.75 \times (1-f_2)$, or \item accepting the remainder of bank $3$'s assets after bank $2$ receives $\$1$ ($f_1 = 0$) to recover $\$0.50$. \end{enumerate*} It is trivial to verify that the first case provides the maximal payout so long as $f_2 < 0.4$.
\item If $f_2 = 0.4$: bank $1$ can recover $\$0.50$ with $f_1 = 0$ or $f_1 = 0.5$; in the former case by recovering the remainder of bank $3$'s assets after bank $2$ takes $\$1$ and in the latter case by receiving $\$1$ itself but paying $\$0.50$ in fees to the miners.  Any other choice of fees has bank $1$ paying more fees without receiving a larger share of the payment than in the $f_1 = 0$ case (if $f_1 < 0.4$) or $f_1 = 0.5$ case (if $f_1 > 0.4)$.
\item If $f_2 > 0.4$: the maximum payout possible for bank $1$ is $\$0.50$ with $f_1 = 0$ as any $f_1 \geq f_2$ costs more in fees than can be recovered by taking the remainder of bank $3$'s assets.
\end{enumerate}
With the best response functions, a pure strategy Nash equilibrium would need to satisfy the fixed point problem $f_1 \in f_1^\dagger(f_2)$ and $f_2 \in f_2^\dagger(f_1)$.  Mathematically, we can easily verify that no such fixed point exists.  Intuitively, beginning from proposed fees of $f_1 = f_2 = 0$, the bank's would sequentially increase their fee levels up to $0.4$ or $0.5$ (depending on the choice taken when the other bank is at fee level $0.4$) before cycling back to a fee level of $0$.  Such a cycle demonstrates that the banks would never settle on a given (pure) bidding strategy.

Though no pure strategy equilibrium exists, by Corollary~\ref{cor:nash-exist}, a mixed strategy equilibrium must exist. As the bidding strategy can be completely characterized by $f_i$, in this setting we will define bank $i$'s bidding strategy (for bank $i = 1,2$) such that the fee level $k/10$ is chosen with probability $p_{ik}\geq 0$ for $k \in 10\times\fcal(10)$ (and so that $\sum_{k = 0}^{10} p_{ik} = 1$).  Herein, we will consider only the symmetric mixed strategies as the banks are identical, i.e., $p := p_1 = p_2$ defines the strategy of both banks.  In fact, it can readily be verified that only exists a single equilibrium $p^\dagger$ with
\begin{align*}
p^\dagger &\approx \left(0.1787 \, , \, 0.0634 \, , \, 0.2392 \, , \, 0.1499 \, , \, 0.3689 \, , \, 0 \, , \, 0 \, , \, 0 \, , \, 0 \, , \, 0 \, , \, 0 \right)^\top.
\end{align*}
With this symmetric equilibrium, both banks have the expected payout of (approximately) $\$0.5447$ which, notably provides only a marginal benefit over a \emph{guaranteed} payout of $\$0.50$ from bidding $f_i = 0$ with 100\% probability (i.e., a pure strategy). 
\end{example}

\subsection{Pareto Optimal Bidding}\label{sec:bidding-pareto}

Within this section we wish to consider the setting in which the bids on the (unsecured) obligations $(1-\mu)L_{ij}$ are made so as to be Pareto optimal.  That is, bids are made so that no single bank can improve its (expected) utility without decreasing the utility of another bank.  As with the Nash equilibrium bidding strategies presented above, each bank $i$ seeks to maximize its expected utility
\[\lim_{t \to \infty} \E[u_i(V_i^t(\bb^\ddagger))]\]
for continuous utility function $u_i$.  As in the prior section, by Proposition~\ref{prop:limit-K}, this utility is taken over $[K^*_i(\bb^\ddagger)]^+$ in the $C^\# \geq \abs{\{(i,j)\in [n]^2 \mid L_{ij} > 0\}}$ setting of Section~\ref{sec:limit}.

Using the notion of linear scalarizations (see, e.g.,~\cite[Chapter 11.2.1]{jahn2011}), a Pareto optimal bidding strategy $\bb^\ddagger = (b_{ij}^\ddagger)_{i,j\in [n]}$ can be found via the single optimization problem:
\begin{equation}\label{eq:pareto}
\bb^\ddagger \in \argmax\left\{\lim_{t \to \infty} \sum_{k \in[n]} w_k \E[u_k(V_k^t(\bb))] \; | \; b_{ij} \in \bcal_{ij}(D,F) \; \text{for all} \; i,j\in[n]\right\},
\end{equation}
for $\bw \in \bbr_{++}^{n+1}$.  This optimization problem provides a Pareto optimal solution as, if bank $i$ has an improved utility under a different bidding strategy, then some other bank $j$ must have a lower utility for $\bb^\ddagger$ to be optimal.
In comparison to the Nash equilibrium setting~\eqref{eq:nash}, only pure strategies (i.e., deterministic) ones are considered for the Pareto optimal bidding strategy~\eqref{eq:pareto}.  This is because: \begin{enumerate*} \item banks would prefer the certainty of a deterministic bidding strategy so that the costs are known at the implementation of the contract and \item adding (independent) random bidding strategies cannot improve the system-wide expected utility (weighted by $\bw$) beyond those given by the pure strategies.\end{enumerate*}

\begin{proposition}\label{prop:pareto}
There exists a Pareto optimal bidding strategy to~\eqref{eq:pareto} (for any $\bw \in \bbr_{++}^{n+1}$) for any blockchain clearing system with finite discretization $D,F \in \N$.
\end{proposition}
\begin{proof}
First, if an optimum exists to~\eqref{eq:pareto} for any $\bw \in \bbr_{++}^{n+1}$ then it is Pareto optimal (see, e.g.,~\cite[Chapter 11.2.1]{jahn2011}).  Therefore, existence follows directly from the finite cardinality of the feasible bids $\abs{\bcal_{ij}(D,F)} = (D(1-\mu)L_{ij}+1)(F+1)$ for any pair of banks $(i,j) \in [n]^2$. 
\end{proof}

\begin{remark}\label{rem:no-competition-2}
Recall from Remark~\ref{rem:no-competition}, if $x_i + \sum_{j\in[n]}[\mu(1-f_R) + (1-\mu)]L_{ji} - \sum_{j\in[n]} L_{ij} \geq 0$ for every bank $i$ (i.e., there are no defaults in the system if all obligations are paid at the zero fee level), then the universal application of the zero fee is both a Nash and Pareto optimal bidding strategy.
\end{remark}

We wish to conclude this section by presenting a four bank system which demonstrates that the Pareto optimal bidding strategies can be nontrivial and complex.  This is undertaken in a setting with a Bernoulli systematic shock in which the expected utilities are taken over the random endowments.\footnote{We choose to consider a systematic shock as that setting was shown to be the most relevant for stress testing for centralized clearing in~\cite{BF18comonotonic}.}  
\begin{example}\label{ex:pareto}
Consider a heterogeneous four bank system in which each bank has, additional, external obligations. We will model this with the introduction of an external or society node with no obligations to any of the banks, but holding interbank assets that the other banks owe it.  We will use the standard notation and denote this society node $0$ and place it first in the order of banks.  In this way we will let $[n] = \{0,1,\dots,n\}$ within this example.  
The (random) assets $\bx$ and nominal liabilities for each bank $i$ are given by
\begin{align*}
\bx &= \begin{cases} (1 \, , \, 3 \, , \, 2 \, , \, 5)^\T &\text{with probability } 25\%, \\ (6 \, , \, 8 \, , \, 7 \, , \, 10)^\T &\text{with probability } 75\%, \end{cases} \qquad \text{and} \qquad
\bL = \left(\begin{array}{ccccc} 3 & 0 & 7 & 1 & 1 \\ 3 & 3 & 0 & 3 & 3 \\ 3 & 1 & 1 & 0 & 1 \\ 3 & 1 & 2 & 1 & 0 \end{array}\right).
\end{align*}
Within the liabilities matrix $\bL$, the external obligations to node $0$ are provided within the first column.
We can view the random assets $\bx$ as being a scenario in which there is a 25\% probability of a shock of size $5$ to all banks.
To simplify the discussion, we will assume $C^\# \geq 16$ herein.  Additionally, we will assume that all obligations are uncollateralized ($\mu = 0$) for a direct comparison with the centralized Eisenberg-Noe system~\cite{EN01}.  In fact, under the centralized clearing system of~\cite{EN01}, only bank 4 is solvent after clearing in the stress scenario; bank 1 is a so-called first order default (as it defaults even if all its counterparties pay in full), bank 2 is a second order default (as it defaults as a consequence of bank 1 defaulting, but assuming all other counterparties pay in full), and bank 3 is a third order default (as it defaults as a consequence of the first and second order defaults only).  Without any stress, bank 1 will still default, but there are no cascading failures.

For the decentralized clearing, we consider the discretization of fees as $F = 40$ so that fee levels are at $2.5\%$ intervals and discretization of bids as $D = 100$ so that bids are placed in ``cents''.
We consider the Pareto optimal problem~\eqref{eq:pareto} in which each bank (and external node $0$) wish to maximize their expected equity, i.e., 
\[u_i(z) = z^+, \quad \text{for all} \quad i\in[n]. 
\]
Due to the application of the utility functions to the terminal cash account (i.e., the positive net worths), this utility is, functionally, equivalent to the risk-neutral setting $u_i(z) = z$.
The scalarization to find the Pareto optimal bidding strategy is taken as $\bw = (0.1 \, , \, 1 \, , \, 1 \, , \, 1 \, , \, 1)^\T$, i.e., such that the expected equity of the external node is weighted at one-tenth that of all banks.
The resulting Pareto optimal solution is:
\begin{align*}
\bd^\ddagger(0\%) &= \left(\begin{array}{ccccc} 3 & 0 & 0 & 0 & 0 \\ 3 & 3 & 0 & 3 & 0 \\ 3 & 1 & 1 & 0 & 1 \\ 3 & 1 & 2 & 1 & 0 \end{array}\right),\quad
\bd^\ddagger(2.5\%) = \left(\begin{array}{ccccc} 0 & 0 & 7 & 1 & 0 \\ 0 & 0 & 0 & 0 & 3 \\ 0 & 0 & 0 & 0 & 0 \\ 0 & 0 & 0 & 0 & 0 \end{array}\right),\quad
\bd^\ddagger(5\%) = \left(\begin{array}{ccccc} 0 & 0 & 0 & 0 & 1 \\ 0 & 0 & 0 & 0 & 0 \\ 0 & 0 & 0 & 0 & 0 \\ 0 & 0 & 0 & 0 & 0 \end{array}\right),
\end{align*}
and $d^\ddagger_{ij}(f) = 0$ for every $i,j$ and $f > 5\%$ with $f \in \fcal(40)$.  
Notably, though we imposed a discretization $D = 100$ on the bid amounts, the Pareto optimal clearing solution is all-or-nothing at every fee level.  In fact, any choice $D \in \N$ satisfying Assumption~\ref{ass:discrete} will result in the same Pareto optimal bidding strategy; this is true even in the limit to the continuum of bids $D \to \infty$.

To conclude this example, consider the net worths and utilities under the Pareto optimal bidding strategy with the Eisenberg-Noe centrally cleared solution.
First, consider the centrally cleared Eisenberg-Noe net worths.  Letting $\bK_{\rm EN}^u,\bK_{\rm EN}^s$ denote the net worths in the unstressed or stressed scenarios respectively,
\begin{align*}
\bK_{\rm EN}^u &\approx (-1 \, , \, 5.4167 \, , \, 5.9167 \, , \, 7.9167)^\T, \\
\bK_{\rm EN}^s &\approx (-6.8108 \, , \, -3.0270 \, , \, -0.3243 \, , \, 1.6216)^\T.
\end{align*}
The payments to society are $11.75$ in the unstressed scenario and $9.3784$ in the stressed scenario.  This leads to a (weighted) expected utility of $15.9586$.
Consider, now, the Pareto optimal clearing.  Letting $\bK^u,\bK^s$ denote the net worths in the unstressed and stressed scenarios respectively,
\begin{align*}
\bK^u &= (-1 \, , \, 5.825 \, , \, 5.975 \, , \, 7.875)^\T, \\
\bK^s &\approx (-6.9205 \, , \, -2.5802 \, , \, -0.3629 \, , \, 2.8145)^\T.
\end{align*}
The payments to society are $11$ in the unstressed scenario and $7.9585$ in the stressed scenario.  This leads to a (weighted) expected utility of $16.4838$; as such the blockchain based clearing does indeed improve the (weighted) utility of the system.
\end{example}

\subsection{Implications for Financial Stability}\label{sec:stability}


In contrast to centralized clearing procedures (e.g., \cite{EN01,RV13}), the blockchain clearing system introduced herein has no central authority.  
The main difficulty in real-life clearing, the one which the fixed recovery of~\cite{RV13} is supposed to model, is the combination of determining the available funds and the resolution of disagreements over seniority of payments. Therefore, a typical modeling approximation is to fix the recovery rate and apply it to the entire payment.
In our blockchain clearing mechanism, these two issues are resolved automatically.
This reduces frictions in the clearing procedure as the blockchain system automates the claims resolution procedure through the seniority structure determined by the bids.  This occurs without assuming, e.g., the pro-rata repayment scheme of~\cite{EN01,RV13}.  In allowing the banks to (decentrally and optimally) determine their desired seniority, the blockchain clearing system endogenizes the effective recovery rate.

Furthermore, centrally cleared systems have an often overlooked risk -- cybersecurity risk.\footnote{\url{https://www.bloomberg.com/news/articles/2021-01-12/what-do-wall-street-leaders-think-is-the-next-big-risk}}$^\text{,}$\footnote{\url{https://www.bloomberg.com/news/articles/2021-02-24/fed-investigating-outage-in-interbank-payment-system}}  This systemic risk has been highlighted by recent ransomware attacks (e.g., the Colonial Pipeline ransomware attack).  However, a decentrally cleared system with an immutable distributed ledger has no single entity that can cause the failure of the entire financial clearing system~\cite{singh2016blockchain}.  As such, this decentralized clearing procedure inherently increases financial stability beyond other points considered in this work.  
Notably, and in contrast to other decentralized clearing procedures (e.g.,~\cite{schaarsberg2018solving,ketelaars2020decentralization}), the blockchain clearing system allows for optimized claims resolution (as opposed to exogenous rules such as pro-rata payment schemes assumed in prior works) and guarantees that all payable obligations are fulfilled through the actions of the miners.  

The optimal bids placed in our blockchain clearing system also act as a heuristic, endogenous, stress testing procedure.  As noted in Remark~\ref{rem:no-competition}, if no stresses are exhibited in the financial system, then the $0\%$ fee bid is both Nash and Pareto optimal.  Therefore, as a first order approximation, if non-zero fees are observed in smart contracts then this indicates positive probability of defaults on some obligations.  As higher fees correspond exactly with higher seniority, intuitively, the greater the counterparty risk, the more fees a bank would be willing to pay to guarantee payment.  Such notions are observed in both Example~\ref{ex:no-pure} and Example~\ref{ex:pareto}.  Thus, the optimal bids can be viewed as a measure of systemic risk as viewed by the banks in the system.

In fact, these optimal bids can be viewed as a ``price'' of transacting with a risky counterparty.  These charges can, therefore, incentivize banks to transact with safer counterparties thus reducing the overall risk of default contagion without introducing moral-hazards such as a bailout scheme to support failing institutions.  Specifically, the centrally cleared systems highlighted in Section~\ref{sec:centralized} with recovery rate $\alpha < 1$ can result in bailout coalitions to rescue a failing institution or else face larger losses and the risk of default contagion.  Optimal bidding in a blockchain clearing system can reduce the risk of default contagion through, e.g., greater recovery than is often assumed on defaults in the model of~\cite{RV13}.  We refer the reader to Example~\ref{ex:pareto} to see the high recovery of assets in case of default by comparing the results therein to the $\alpha = 1$ setting of \cite{EN01}.  Therefore the blockchain system dually increases the (short term) recovery of defaulted assets and increases bank accountability for their risky transactions.  That is, the consequences of the blockchain system can lead to greater financial stability over the centrally cleared system; though beyond the scope of this work, a comparison of the network formation problem under central and blockchain cleared systems would be of great interest to fully understand these implications.


%

\section{Conclusion}\label{sec:conclusion}
In this paper we proposed a novel decentralized clearing mechanism through blockchain, which endogenously and automatically provides a claims resolution procedure. We provided an algorithm which builds the entire blockchain to clear these obligations so as to guarantee the payments can be verified and the miners earn a fee. We show that, conditional on blocks being sufficiently large, there always exist greatest and least clearing net worth. Moreover, under weak regularity conditions, there is a unique clearing net worths vector.  We also considered the formulation of the optimal bidding strategies for each firm in the network so that all firms are utility maximizers with respect to their terminal wealths.  These optimal bidding strategies are approached in both Nash and Pareto settings.  

This analysis of the blockchain clearing system can be extended in the following ways.  
First, more detailed and complex decentralized clearing rules could be considered.  For example, verifying transactions costs the miners computational effort. The fees are a way to compensate the miners for this pivotal task.  As such, a reserve price $\$\epsilon > 0$ can be imposed so that the collected fees need to exceed this threshold for the miners to even consider placing the payment on the blockchain.  This would imply that $\bcal_t((\widehat\bb^s)_{s < t})$ is appended to include the additional constraint $\int_0^1 f b_{ij}^t(f)df \in \{0\} \cup [\epsilon,\infty)$ for every $i,j \in [n]$.  Such a constraint guarantees that the blockchain will clear the network of obligations to converge to its terminal state in a finite number of blocks (with the possibility that some defaulting banks have a small, but positive, terminal cash account).  Such a reserve price, were it known, would also influence the bidding strategies implemented by the banks as well.
Second, we implicitly assumed that the optimal bidding strategies were determined under a complete information setting (i.e., each bank has full information about the network and the bidding implemented by all other banks).  However, a partial information system could be considered instead (e.g., banks only have local information about the liability network).  Under such a setting, each bank would have (heterogeneous) beliefs about the state of the system. An optimal bidding strategy implemented by bank $i$ would then be taken conditional on the information set available to it. Such a setting presents difficulties in modeling the belief set for each bank that needs to be considered more fully.
Finally, as begun within Section~\ref{sec:stability}, implications of this blockchain clearing system for systemic risk is of great importance.  For instance, we recommend a replication of~\cite{AOT15,amini-compression} on the implications of network topology on financial stability for this decentralized clearing system. 


\bibliographystyle{plain}
\bibliography{bibtex2}

\end{document}